\documentclass[superscriptaddress,aps,pra,twocolumn,showpacs,nofootinbib,longbibliography]{revtex4-1}
\usepackage{etex}
\usepackage{amsmath,amssymb,amsthm}
\usepackage{easybmat}
\usepackage[colorlinks=true,citecolor=blue,urlcolor=blue]{hyperref}
\usepackage[pdftex]{graphicx}
\usepackage{times,txfonts}
\usepackage{braket}
\usepackage{color}
\usepackage{natbib}

\newcommand{\bla}{\color{black}}
\newcommand{\be}{\begin{equation}}
\newcommand{\ee}{\end{equation}}
\newcommand{\ba}{\begin{eqnarray}}
\newcommand{\ea}{\end{eqnarray}}
\newcommand{\ketbra}[2]{|#1\rangle \langle #2|}
\newcommand{\tr}{\operatorname{Tr}}

\newcommand{\half}{\frac{1}{2}}

\newcommand{\etal}{{\it{et. al. }}}

\newtheorem{thm}{Theorem}

\begin{document}
\title{Nonclassicality  of  local bipartite  correlations}  \author{C.
  Jebaratnam}  \email{jebarathinam@gmail.com} \affiliation{S.  N. Bose
  National Center  for the Basic Sciences,  Kolkata, India} 
  \author{S. Aravinda} 
\affiliation{Institute of Mathematical Sciences, Chennai, India}
\author{R.   Srikanth} \affiliation{Poornaprajna Institute
  of Scientific Research, Bangalore, Karnataka, India}

\begin{abstract}
Simulating  quantum  nonlocality   and  steering  requires  augmenting
pre-shared  randomness with  non-vanishing  communication cost.   This
prompts  the question  of  how  one may  provide  such an  operational
characterization  for  the quantumness  of  correlations  due even  to
unentangled states.  Here we show that  for a certain class of states,
such  quantumness can  be pointed  out by  \textit{superlocality}, the
requirement for  a larger  dimension of  the pre-shared  randomness to
simulate  the  correlations  than  that  of  the  quantum  state  that
generates   them.    This  provides   an   approach   to  define   the
nonclassicality   of  local   multipartite   correlations  in   convex
operational theories.
\end{abstract}

\maketitle

\section{Introduction}

Local measurements on a spatially separated quantum system can lead to
a nonclassical box (set of  correlations) which cannot be explained by
shared  classical  randomness  \cite{Bel64,BCP+14}.  This  feature  of
quantum correlations  is called  nonlocality and  is witnessed  by the
violation  of  a Bell  inequality,  which  must  be satisfied  by  the
correlations  that   admit  a   local  hidden  variable   (LHV)  model
\cite{Bel64}.   The fact  that the  nonlocality of  quantum theory  is
limited \cite{Cir80} led Popescu  and Rohrlich to propose nonsignaling
correlations  which  are  more   nonlocal  than  quantum  correlations
\cite{PR94,Jeba}.    One  of   the  goals   of  studying   generalized
non-signaling  probability  theories  is  to find  out  what  physical
principles limit quantum nonlocality \cite{Pop14}.

Concepts  like  quantum   discord  \cite{OZ01,HV01,MBA+12}  and  local
broadcasting \cite{PHH08}  indicate the existence of  quantumness even
in separable states, and can  be associated with the non-commutativity
of measurements  \cite{Guo16}.  It  is known  that the  observation of
nonlocality or Einstein-Podolsky-Rosen (EPR) steering also implies the
presence  of incompatibility  of  measurements  \cite{Ban15}. From  an
operational   perspective,  nonlocal   or  steerable   states  require
augmenting  pre-shared  randomness  with non-zero  communication  cost
\cite{TB03,SAB+16}.

Here  we are  concerned  with the  question  of how  to  give such  an
operational characterization to the  quantumness of local correlations
arising from non-commuting measurements performed on separable states.
By definition, such a box clearly requires zero communication cost.

We partially answer this question  by providing evidence that for some
such states  in the bipartite two-input-two-output  Bell scenario, the
dimension of  the pre-shared randomness  required to simulate  the box
exceeds the dimension of the quantum  system generating it.  This is a
specific   case  of   superlocality   \cite{DW15}.    The  idea   that
superlocality occurs even in separable states implicitly finds mention
in  \cite{GBS16} (cf.   in  particular, Fig.   3  there).  A  detailed
characterization     of     superlocality    for     the     bipartite
single-input-multiple-output   scenario  appears   in  Sec.    4.1  of
\cite{Zha12}. Bounds  on the  quantum dimension required  to reproduce
Bell  correlations  in  the  bipartite  multiple-input-multiple-output
scenario are discussed in \cite{WS16}.

\section{The polytope of nonsignaling boxes}

In  the  formalism  of   generalized  no-signaling  theory,  bipartite
correlations are treated as ``boxes'' shared between two parties.  Let
us denote the input variables on  Alice's and Bob's sides $x$ and $y$,
respectively,  and the  outputs $a$  and  $b$, as  depicted in  Figure
\ref{fig:scenario}.  We restrict ourselves to the state space in which
the  boxes  have two  binary  inputs  and  two binary  outputs,  i.e.,
$x,y,a,b \in \{0,1\}$.  In this case,  the state of every box is given
by   the   set   of   $16$   conditional   probability   distributions
$P(ab|A_xB_y)$.
\begin{figure}
\includegraphics[width=6cm]{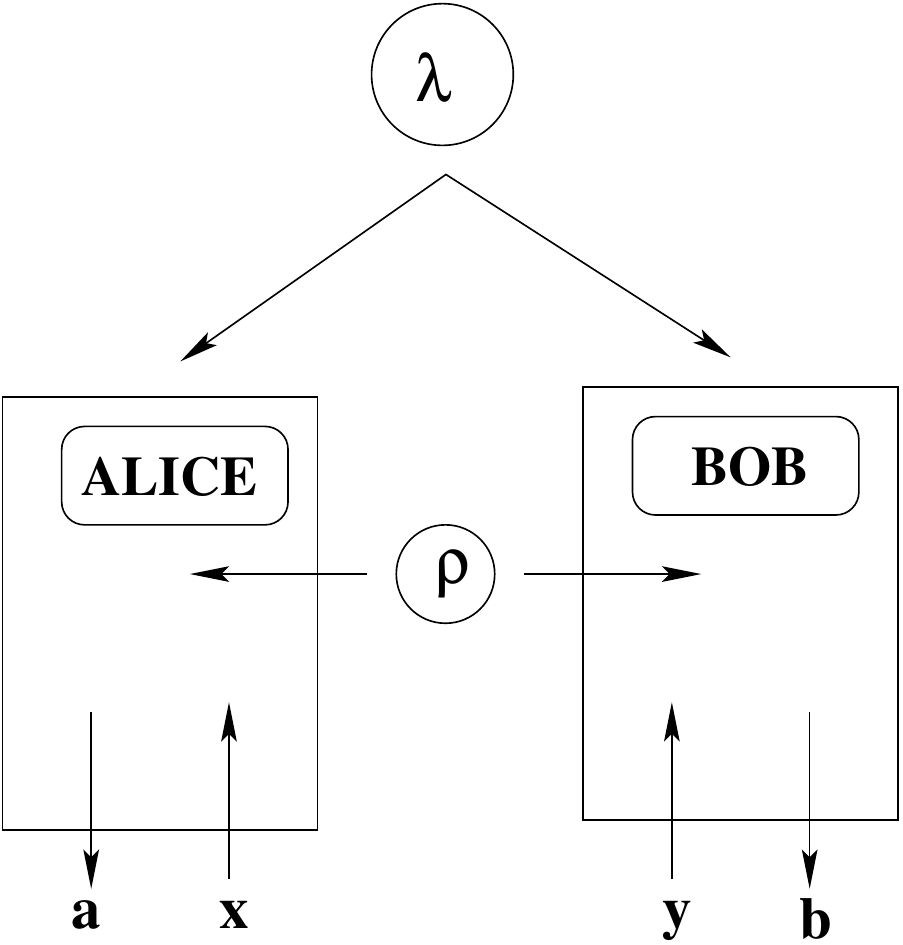}
\caption{Bell-CHSH  2-input  2-output  scenario  appropriate  for  the
  correlations  considered here,  where $x$  and $y$  are Alice's  and
  Bob's  inputs, respectively;  $a$  and $b$  are  their outputs;  and
  $x,y,a,b \in \{0,1\}$. Local quantum  state $\rho$ is shared between
  them, giving  rise to  probability $P(ab|xy)$.  The  simulability of
  this  probability by  pre-sharing  classical  variable $\lambda$  of
  bounded dimension is considered here.}
\label{fig:scenario}
\end{figure}

Barrett  \etal  \cite{BLM+05} showed  that  the  set $\mathcal{N}$  of
two-input-two-output non-signaling (NS) boxes forms an $8$ dimensional
convex  polytope with  $24$ extremal  boxes, the  $8$ Popescu-Rohrlich
(PR) boxes:
\begin{align}
&P^{\alpha\beta\gamma}_{PR}(ab|A_xB_y)
\\&=\left\{
\begin{array}{lr}
\frac{1}{2}, & a\oplus b=x\cdot y \oplus \alpha x\oplus \beta y \oplus \gamma\\ 
0 , & \text{otherwise}\\
\end{array}
\right. \label{NLV}
\end{align}
and $16$ local-deterministic boxes
\begin{equation}
P^{\alpha\beta\gamma\epsilon}_D(ab|A_xB_y)=\left\{
\begin{array}{lr}
1, & a=\alpha x\oplus \beta\\
   & b=\gamma y\oplus \epsilon \\
0 , & \text{otherwise}.\\
\end{array}
\right.   
\label{eq:locdet}
\end{equation}   
Here, $\alpha,\beta,\gamma,\epsilon\in  \{0,1\}$ and  $\oplus$ denotes
addition modulo  $2$.  All the  deterministic boxes can be  written as
the   product   of  marginals   corresponding   to   Alice  and   Bob,
$P_D(ab|A_xB_y)=P_D(a|A_x)P_D(b|B_y)$, whereas the $8$ PR-boxes cannot
be written  in product  form or  even a  convex combination  over such
product boxes. The marginals of the PR boxes are maximally mixed, i.e.
$P(a|A_x)=\frac{1}{2}=P(b|B_y)$ for all $x,y,a,b$.

The polytope $\mathcal{N}$ can be  divided into two disjoint sets: the
local  polytope,  which   is  the  convex  hull  of   the  $16$  local
deterministic  boxes (\ref{eq:locdet}),  and its  nonlocal complement.
The  extremal  boxes in  a  given  set  are equivalent  under  ``local
reversible operations''  (LRO).  LRO  is given  by Alice  changing her
input $x\rightarrow x\oplus 1$, and changing her output conditioned on
the input: $a\rightarrow a\oplus\alpha x\oplus\beta$.  Bob can perform
similar operations.

Fine \cite{Fin82}  showed that  a box has  a LHV model  iff it  can be
written in the  above form. A local box satisfies  the complete set of
Bell-type   inequalities  \cite{WW01a}.    The  Bell-CHSH   inequality
\cite{CHS+69} and its symmetries which are given by,
\begin{align}
&\mathcal{B}_{\alpha\beta\gamma} := (-1)^\gamma\braket{A_0B_0}+(-1)^{\beta \oplus \gamma}\braket{A_0B_1}\nonumber\\
&+(-1)^{\alpha \oplus \gamma}\braket{A_1B_0}+(-1)^{\alpha \oplus \beta \oplus \gamma \oplus 1} \braket{A_1B_1}\le2, 
\label{BCHSH}
\end{align}
form the  complete set,  where $\braket{A_xB_y}=\sum_{ab}(-1)^{a\oplus
  b}P(ab|A_xB_y)$.  All  these tight  Bell-type inequalities  form the
nontrivial facets for  the local polytope.  All  nonlocal boxes, which
lie outside the local polytope, violate a Bell-CHSH inequality.

Quantum boxes which belong to the Bell-CHSH scenario \cite{CHS+69} are
obtained obtained by two  dichotomic measurements on bipartite quantum
states    described    by    the    density    matrix    $\rho_{AB}\in
\mathcal{B}(\mathcal{H}_A\otimes\mathcal{H}_B)$,  the  set of  bounded
operators in the joint Hilbert space  of the two particles. The Born's
rule predicts the behavior of the quantum boxes as follows:
\begin{equation}
P(ab|A_xB_y)   =    \mathrm{Tr}   \left(\rho_{AB}\mathcal{M}_{A_x}^{a}
\otimes \mathcal{M}_{B_y}^{b}\right),\label{QNS}
\end{equation}
where  $\mathcal{M}_{A_x}^{a}$  and  $\mathcal{M}_{B_y}^{b}$  are  the
measurement operators generating binary outcomes.  A box is quantum if
it can  be written in the  above form.  Otherwise, it  is non-quantum.
Quantum nonlocal  boxes in  the two-input-two-output Bell  violate the
Bell-CHSH inequality up to $2\sqrt{2}$ \cite{Cir80}.

The set  of quantum  boxes over all  input-output scenarios  is convex
\cite{WW01}. With unrestricted dimension, any local box can be written
in the form given in Eq.  (\ref{QNS}).  The set $\mathcal{L}$ of local
boxes and the set $Q$ of  quantum boxes satisfy $\mathcal{L} \subset Q
\subset \mathcal{N}$.  But note that with dimensional restriction, the
set of quantum correlations is  not in general convex \cite{DW15}.  

     Noting  that quantum  discord \cite{OZ01,HV01,MBA+12}  provides a
measure of  quantum correlations  going beyond nonlocality,  a natural
question   that   arises   here   concerns   its   relationship   with
superlocality, an issue which we address in Section \ref{conc}.

In  this work,  we  characterize correlated  boxes  arising from  spin
projective measurements  $A_x=\hat{a}_x \cdot \vec{\sigma}$  and $B_y=
\hat{b}_y\cdot  \vec{\sigma}$  along  the directions  $\hat{a}_x$  and
$\hat{b}_y$ on two-qubit systems.   Here, $\vec{\sigma}$ is the vector
of Pauli matrices.

\section{Bell nonclassicality of local boxes}

The generation of  local boxes using quantum  systems requires neither
entanglement nor  noncommutativity.  This  follows from the  fact that
such local  boxes can be classically  simulated without communication.
However,  when  the local  Hilbert space  dimensions  of the  measured
system are restricted,  the quantum simulation of  certain local boxes
may  require  both noncommuting  measurements  and  states which  have
quantum correlations.

The noisy PR-box  \cite{MAG06} is a mixture of  a PR-box and white
noise,
\begin{equation} 
P=p_{PR}P_{PR}+(1-p_{PR})P_N,
\label{PRiso} 
\end{equation} 
where $p_{PR}$ is a real number such that $0\le p_{PR}\le 1$, $P_{PR}$
denotes the  PR-box $P^{000}_{PR}$, and  $P_N$ is the  maximally mixed
box,  i.e., $P_N(ab|A_xB_y)=1/4$  for  all  $x,y,a,b$.  The  noisy
PR-box      violates     the      Bell-CHSH     inequality,      i.e.,
$\mathcal{B}_{000}=4p_{PR}>2$     iff    $p_{PR}>\frac{1}{2}$.      If
$p_{PR}>1/\sqrt{2}$,  in   (\ref{PRiso}),  then   $\mathcal{B}_{000}  >
2\sqrt{2}$   in  violation   of  the   Tsirelson  bound   for  quantum
correlations, and hence is physically prohibited.

We discuss two methods by which the box (\ref{PRiso}) can be generated
quantum  mechanically.   In  both   cases,  we  use  the  noncommuting
measurements   $A_0   =   \sigma_x$,   $A_1  =   \sigma_y$,   $B_0   =
\frac{1}{\sqrt{2}}(\sigma_x-\sigma_y)$        and        $B_1        =
\frac{1}{\sqrt{2}}(\sigma_x+\sigma_y)$.

In the first method, the above  measurements are applied to the family
of pure entangled states:
\begin{equation}
\ket{\psi(\theta)}  =  \cos\theta\ket{00}+\sin\theta\ket{11}; \quad  0
\le \theta \le \pi/4, \label{nmE}
\end{equation}
which produces the noisy CHSH box:
\begin{equation}
P^\mathcal{C}_{CHSH}       =       \frac{2+(-1)^{a\oplus       b\oplus
    xy}\sqrt{2}\mathcal{C}}{8}. \label{chshfam}
\end{equation}
Here $\mathcal{C}=\sin2\theta$, the  concurrence \cite{Woo98} of state
(\ref{nmE}).  The  above statistics  can be  written in  the noisy
PR-box form with $p_{PR}=\mathcal{C}/\sqrt{2}$.  So, $p_{PR}>0$ if and
only if the state is entangled $(\theta>0$).  This holds even when $P$
becomes local ($\theta < \frac{\pi}{8}$).

In the second  method to generate the box  (\ref{PRiso}), consider the
two-qubit Werner states,
\begin{equation}
\rho_W=W\ketbra{\psi^+}{\psi^+}+(1-W)\frac{\openone}{4}, 
\label{eq:werner}
\end{equation}
where $\ket{\psi^+}=\frac{1}{\sqrt{2}}(\ket{00}+\ket{11})$.  The above
states are entangled iff $W>\frac{1}{3}$ \cite{Wer89} and nonlocal iff
$W>\frac{1}{\sqrt{2}}$.   It  is known  that  the  Werner states  have
nonzero  quantumness   (as  quantified  by  discord)   for  any  $W>0$
\cite{OZ01}.  For  the noncommuting  measurements that we  used above,
the Werner state (\ref{eq:werner}) gives  rise to the noisy PR-box
(\ref{PRiso}) with $p_{PR}= \frac{W}{\sqrt{2}}$.

Even in the local range $0  < p_{PR} \le \half$, the box (\ref{PRiso})
cannot be reproduced  by a pre-sharing just one bit  each of classical
random  correlation.   To  see  this,  note  that  the  noisy PR box
corresponds to the following correlations:
\begin{subequations}
\begin{eqnarray}
&& \langle  A_0\rangle  = \langle  A_1\rangle  =  \langle B_0\rangle  =
\langle B_1\rangle = 0. \label{eq:qcorra}\\ 
&& \langle A_0B_0 \rangle = \langle
A_0B_1 \rangle=  \langle A_1B_0  \rangle =  -\langle A_0B_0  \rangle =
p_{PR}.
\label{eq:qcorrb}
\end{eqnarray}
\label{eq:qcorr}
\end{subequations}
Quite generally, suppose that  the pre-shared bit $\lambda$ determines
the following indeterministic strategy.  Alice outputs $a$ conditioned
on  input   $x$  and  pre-shared  value   $\lambda$  with  probability
$P_{A,\lambda}(a|x)$,  and similarly  Bob outputs  $b$ conditioned  on
input $y$  and $\lambda$  with probability  $P_{B,\lambda}(b|y)$.  The
value  of  $\lambda  \in  \{0,1\}$ is  distributed  according  to  the
probability distribution $P_\lambda(\lambda)$.
Eq. (\ref{eq:qcorra}) implies:
\begin{align}
P_{A,0}(0|0)P_\lambda(0) + P_{A,1}(0|0)P_\lambda(1) &= \half \nonumber \\
P_{A,0}(0|1)P_\lambda(0) + P_{A,1}(0|1)P_\lambda(1) &= \half \nonumber \\
P_{B,0}(0|0)P_\lambda(0) + P_{B,1}(0|0)P_\lambda(1) &= \half \nonumber \\
P_{B,0}(0|1)P_\lambda(0) + P_{B,1}(0|1)P_\lambda(1) &= \half.
\label{eq:qcorra+}
\end{align}
Eq.  (\ref{eq:qcorrb}) implies:
\begin{align}
P_{A,0}(0|0)P_{B,0}(0|0)P_\lambda(0) &+ P_{A,1}(0|0)P_{B,1}(0|0)P_\lambda(1) = 
\frac{1+p}{4}\nonumber \\
P_{A,0}(0|1)P_{B,0}(0|0)P_\lambda(0) &+ P_{A,1}(0|1)P_{B,1}(0|0)P_\lambda(1) =
\frac{1+p}{4} \nonumber \\
P_{A,0}(0|0)P_{B,0}(0|1)P_\lambda(0) &+ P_{A,1}(0|0)P_{B,1}(0|1)P_\lambda(1) =
\frac{1+p}{4}\nonumber \\
P_{A,0}(0|1)P_{B,0}(0|1)P_\lambda(0) &+ P_{A,1}(0|1)P_{B,1}(0|1)P_\lambda(1) =
\frac{1-p}{4},
\label{eq:qcorrb+}
\end{align}
where we  make use of  the normalization for the  relevant conditioned
probabilities.

Now,  subtracting the  first two  equations of  (\ref{eq:qcorra+}), we
find:
\begin{align}
(P_{A,0}(0|0)  &-P_{A,0}(0|1))P_\lambda(0) + \nonumber\\
&(P_{A,1}(0|0)-P_{A,1}(0|1))P_\lambda(1) = 0,
\label{eq:qcorradiff}
\end{align}
while subtracting  the first  two equations of  (\ref{eq:qcorrb+}), we
find:
\begin{align}
(P_{A,0}(0|0) &-P_{,0}A(0|1))P_{B,0}(0|0)P_\lambda(0) + \nonumber\\
&(P_{A,1}(0|0)-P_{A,1}(0|1))P_{B,1}(0|0)P_\lambda(1) = 0.
\label{eq:qcorrbdiff}
\end{align}
From   Eqs.   (\ref{eq:qcorradiff})   and  (\ref{eq:qcorrbdiff}),   we
determine that $P_{B,0}(0|0) = P_{B,1}(0|0)$.   Plugging this in the first
equation  of (\ref{eq:qcorra+}),  we find  $P_{B,0}(0|0) =  P_{B,1}(0|0) =
\half$.

Proceeding     thus    for     other    conditional     probabilities,
$P_{A,\lambda}(a|x)$   and  $P_{B,\lambda}(b|y)$,   we  derive   their
measurement independence  on the  underlying pre-shared  variable, and
their  value  to  be  $\half$.   Substituting  these  values  for  the
conditional   probabilities    in   the   first   equation    of   Eq.
(\ref{eq:qcorrb+}),  we  find  $p_\lambda(0)  +  p_\lambda(1)=1+p$,  a
contradiction for any  $p>0$.  This entails that the  dimension of the
classical  system  simulating  the   noisy PR  box  must  exceed  the
dimension   (two)   of   the   qubit,    and   is   an   instance   of
\textit{superlocality} \cite{DW15}.

This   observation   prompts   us  to   operationally   identify   the
nonclassicality of  the box  (\ref{PRiso}) with  superlocality.  Since
this characterization of  nonclassicality depends only on  the box and
not how it is generated, our  approach gives a general way to approach
nonclassicality  in local  correlations  in  an arbitrary  operational
theory.  For noisy PR boxes, we identify $p_{PR}$ as a nonclassicality
measure. Our  result means that all  points in the diagonal  in Figure
\ref{fig:nonconv}, even  those with  the local region  \textbf{L}, are
superlocal.
\begin{figure}
\includegraphics[width=6cm]{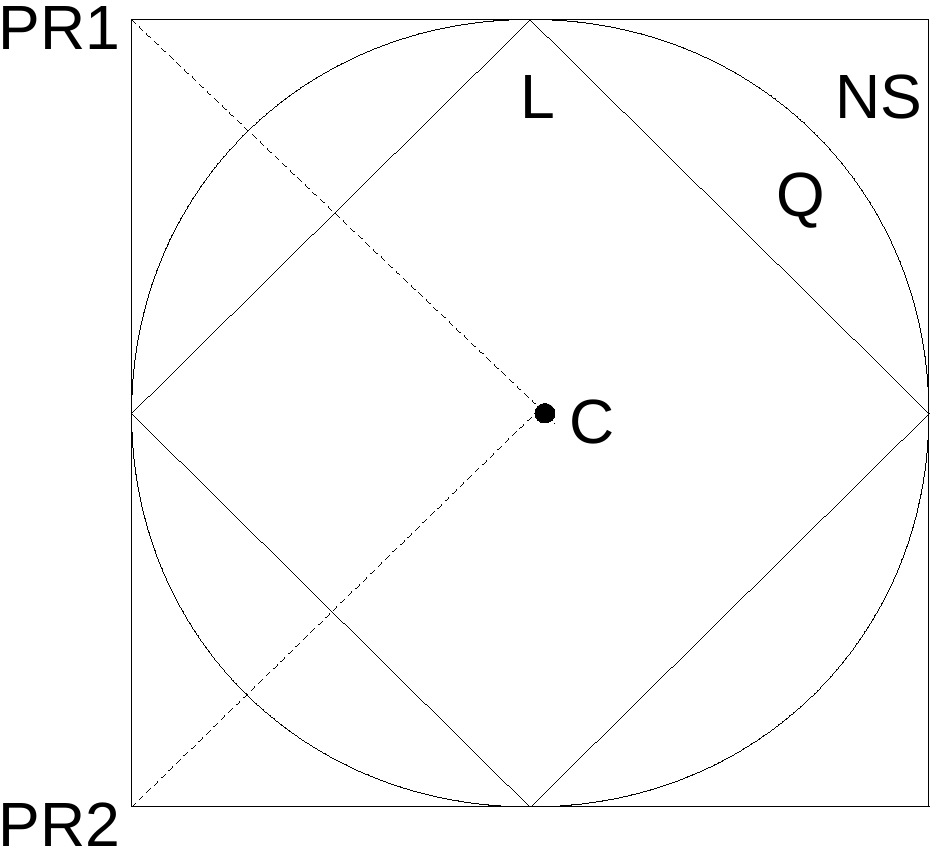}
\caption{A section of the correlations  that appear in the scenario of
  Figure \ref{fig:scenario}.  The region  \textbf{NS} within the outer
  square,  bounded  by  PR   boxes  \textbf{PR1},  \textbf{PR2},  etc.
  represents  nonsignaling correlations,  with the  quantum realizable
  ones given by  the circular region \textbf{Q}, a subset  of which is
  the local polytope \textbf{L}.  The center \textbf{C} corresponds to
  the white  noise box.  In  each dashed diagonal line  radiating from
  \textbf{C} towards  a PR box  but excluding \textbf{C},  the segment
  within \textbf{Q}  represents a family  of quantum boxes  having the
  form   of  noisy   PR-boxes,  which   are  superlocal   even  within
  \textbf{L}.}
\label{fig:nonconv}
\end{figure}

\section{Entropic superlocality\label{BD}}

The  noisy   local   CHSH   box (\ref{chshfam}) can be decomposed in 
terms of the local deterministic boxes alone as follows:
\begin{align}
P^{\mathcal{C}=1/\sqrt{2}}_{CHSH}&=\frac{1}{8}\sum_{\alpha\beta\gamma}P^{\alpha\beta\gamma(\alpha\gamma\oplus\beta)}_D(ab|xy) \\
 &= \frac{1}{4}\left(\frac{P^{0000}_D+P^{1000}_D}{2}\right)+\frac{1}{4}\left(\frac{P^{0010}_D+P^{1110}_D}{2}\right) \nonumber \\
&+ \frac{1}{4}\left(\frac{P^{0101}_D+P^{1101}_D}{2}\right)+\frac{1}{4}\left(\frac{P^{0111}_D+P^{1011}_D}{2}\right) \label{eq:simula}\\
&\equiv \frac{1}{4}\left(\Delta_1 
+ \Delta_2 + \Delta_3 + \Delta_4\right) \nonumber\\
 &=\frac{1}{2}P^{000}_{PR}+\frac{1}{2}P_N.  
\end{align} 
Accordingly,  for   $p_{PR}\le1/2$,  the  noisy  PR   box  can  be
decomposed as follows:
\begin{align}
&p_{PR}P^{000}_{PR}+(1-p_{PR})P_N \nonumber \\
&=
\frac{2p_{PR}}{8}\sum_{\alpha\beta\gamma}P^{\alpha\beta\gamma(\alpha\gamma\oplus\beta)}_D(ab|xy)+
(1-2p_{PR})P_N, 
\nonumber \\
 &=\frac{2p_{PR}}{4}\sum_{j=1}^4 \Delta_j(ab|xy) + (1-2p_{PR})P_N. 
\label{eq:bellstr-} \\ 
&=\frac{1}{4}\sum_{\lambda=1}^{4}P_\lambda(a|A_x)P_\lambda(b|B_y),
\label{eq:bellstr}
\end{align}
where one of the parties  (here, Alice) uses nondeterministic  
strategies given by:
\begin{align}
P_1(a|A_x)&=(P^{00}_D+2p_{PR}P_D^{10}+(1-2p_{PR})P_D^{01})/2, \nonumber \\
P_2(a|A_x)&=(P^{00}_D+2p_{PR}P_D^{11}+(1-2p_{PR})P_D^{01})/2, \nonumber \\
P_3(a|A_x)&=(P^{01}_D+2p_{PR}P_D^{11}+(1-2p_{PR})P_D^{00})/2, \nonumber \\
P_4(a|A_x)&=(P^{01}_D+2p_{PR}P_D^{10}+(1-2p_{PR})P_D^{00})/2,
\end{align}
while the other (here, Bob) uses deterministic strategies given by:
\begin{align}
&P_1(b|B_y)=P^{00}_D, P_2(b|B_y)=P^{10}_D, \nonumber \\
&P_3(b|B_y)=P^{01}_D, P_4(b|B_y)=P^{11}_D.
\end{align} 

The  expression (\ref{eq:bellstr})  determines a  classical simulation
protocol  with dimension  four, which  is known  to suffice  for local
polytope  in   the  Bell-CHSH   scenario  \cite{DW15}.   We   can  use
(\ref{eq:bellstr-}) to define a  classical communication protocol that
bounds  from  above the  average  pre-shared  information required  to
simulate an noisy PR box.

Assume that  we wish to simulate  an experiment with $n$  trials, with
sufficiently large $n$.  Alice and Bob pre-share a five-symbol string,
say with  symbols $\sigma=0,1,2,3,4$,  such that  the 0's  will appear
with probability  $1-2p_{PR}$ and determine coordinates  where each of
them  independently outputs  unbiased random  bits, when  given either
input.  The remaining  $\sigma$  values,  uniformly distributed,  will
determine  when they  will use  one of  the above  local probabilistic
strategies $\Delta_j$.

In other words,  we require a Shannon encoding for  a source with five
symbols determined  by the probability  distribution $\left(1-2p_{PR},
\frac{p_{PR}}{2},          \frac{p_{PR}}{2},         \frac{p_{PR}}{2},
\frac{p_{PR}}{2}\right)$,   where  $p_{PR}\le\half$.    Therefore,  on
average,  per trial  Alice and  Bob must  pre-share $l(p_{PR})$  bits,
where
\begin{equation}
l(x) \equiv (1-2x)\log(1-2x) + 2x \log\left(\frac{x}{2}\right).
\label{eq:L}
\end{equation} 
We  find that  $l(x)  \ge  1$ for  $x\gtrsim  0.085$.  Therefore,  the
noisy PR box may be  considered, on average, entropically superlocal,
and thus nonclassical, in the range $0.085 \lesssim p_{PR} \le 0.5$.

\section{Towards quantifying superlocality}

We  now  propose  a  measure   to  quantify  superlocality,  which  is
constructed  to work  for noisy  PR  boxes.        This would  provide
insight on  connecting superlocality  to measures of  quantumness that
would not  necessarily vanish  for separable  states, such  as quantum
discord,  and  that  may  potentially lead  to  quantify  measures  of
correlations   going  beyond   nonlocality,   applicable  to   general
non-signaling correlations.

Later below we will discuss our new measure's limitations when applied
to other  boxes.  Essentially, we  require a quantification of  the PR
box  fraction in  a noisy  PR  box that  would be  independent of  the
particular PR box.  We call this measure ``Bell strength'', because it
employs the Bell correlator.

Define  the  absolute  Bell functions  $\mathcal{B}_{2\alpha+\beta}  =
|\braket{A_0B_0}        +         (-1)^{\beta}\braket{A_0B_1}        +
(-1)^{\alpha}\braket{A_1B_0}  + (-1)^{\alpha  \oplus  \beta \oplus  1}
\braket{A_1B_1}|$.  We construct the following triad of quantities
\begin{eqnarray}
\Gamma_1&:=&\tau\left(\mathcal{B}_0,\mathcal{B}_1,\mathcal{B}_2,\mathcal{B}_3\right)\nonumber\\
\Gamma_2&:=&\tau\left(\mathcal{B}_0,\mathcal{B}_2,\mathcal{B}_1,\mathcal{B}_3\right) \nonumber\\
\Gamma_3&:=&\tau\left(\mathcal{B}_0,\mathcal{B}_3,\mathcal{B}_1,\mathcal{B}_2\right),\label{gi}
\end{eqnarray}
where    
\begin{equation}
\tau\left(\mathcal{B}_0,   \mathcal{B}_1,    \mathcal{B}_2,
\mathcal{B}_3\right)  \equiv \Big||\mathcal{B}_0  - \mathcal{B}_1  | -
|\mathcal{B}_2  - \mathcal{B}_3|\Big|,
\end{equation}  
and  so on.   There are  24 permutational  possibilities for  function
$\tau$, but because of the  three two-fold symmetries $\tau(a,b,c,d) =
\tau(b,a,c,d)=\tau(a,b,d,c)=\tau(c,d,a,b)$,   there   is   an   8-fold
redundancy,  so  that  only  three   terms  $\Gamma_j$,  as  given  in
Eq. (\ref{gi}), are independent.

Here $\Gamma_i$ are  constructed such that it  satisfies the following
axiomatic  properties: (a)  $\Gamma_i\ge0$; (b)  $\Gamma_i=0$ for  any
$P_D^{\alpha\beta\gamma\epsilon}$;  (c)  the  $\Gamma_i$  attains  its
maximum (of 4) on PR-boxes.

We  define Bell strength as:
\begin{equation}
\Gamma := \min_i \Gamma_i. 
\label{eq:G}
\end{equation}
The  quantity $\Gamma$  is  manifestly  LRO invariant.   Further,
$\Gamma(P)=4p_{PR}$ for the noisy PR box (\ref{PRiso}).

Any noisy PR  box, $P=pP^{\alpha\beta\gamma}_{PR}+(1-p)P_N$ with $p>0$
has the property that only one of the Bell functions is nonzero.  This
follows from the fact that $\mathcal{B}_j(P_N)=0$ for any $j$, and the
property of \textit{monoandry}, described below.

Given a  no-signaling correlation shared by  three particles, monogamy
\cite{Ton06}  refers to  a  bound  on the  sum  of  the absolute  Bell
function values for two different  pairs of particles, with respect to
any  fixed  Bell  operator  (say  $\mathcal{B}_{0}$).   In  contrast,
monoandry refers to a bound on the  sum of the absolute values for two
different     Bell    operators     (say    $\mathcal{B}_{0}$     and
$\mathcal{B}_{3}$), with respect to the same pair of particles.

Any given PR  box has a tight association with  the Bell functions, in
that it takes the  value 4 on precisely one of  the four absolute Bell
functions,  and vanishes  for  the rest.   For each  of  the 16  local
deterministic boxes, the  absolute Bell function takes  the value $2$.
This leads to the following monoandry relation:
\begin{equation}
\mathcal{B}_j(P) + \mathcal{B}_k(P) \le 4,
\label{eq:bellandr}
\end{equation}
where $j\ne  k$ and $j,  k \in  \{0,1,2,3\}$.  To prove  this, let
$\mathcal{B}_j^\ast$   denote  $\mathcal{B}_j$   without  taking   the
absolute   value.    Consider    the   decomposition   $\textbf{P}   =
\sum_{j=0}^{3} G_{j}  P^{j\pm}_{\rm PR} + (1-G)L_{\rm  Bell}$, where
$G = \sum_j  G_j$, and $P^{j\pm}_{\rm PR}$ is precisely  one of the PR
box/antibox    pair    $(P_{PR}^{j+},    P_{PR}^{j-})$    such    that
$\mathcal{B}_j^\ast\left(P^{j\pm}_{\rm     PR}\right)=\pm4$.      This
decomposition always  holds, since an equal  mixture a PR box  and its
antibox is the  maximally mixed state $P_N$, which can  be included in
$L_{\rm  Bell}$, the  local  box. 

Consider  any   two  distinct   $\mathcal{B}_j$'s.  Without   loss  of
generality, let these correspond to $j=0,3$. We then have:
\begin{equation}
\mathcal{B}_{0}^\ast(P) = \pm 4G_{0} \pm (1-G)\mathcal{B}_{0}^\ast(L_{\rm Bell}),
\end{equation}
which implies that
\begin{equation}
\mathcal{B}_{0}(P) \le 4G_{0} + 2(1-G).
\end{equation}
Similarly,
\begin{equation}
\mathcal{B}_{3}(P) \le 4G_{3} + 2(1-G),
\end{equation}
from which it follows that
\begin{align}
\mathcal{B}_{0}(P) + \mathcal{B}_{3}(P) &\le 4(G_{0}+G_{3}) + 4(1-G)\le 4,
\end{align}
since $G_{0}+G_{3}\le G$. Clearly, this  holds for any distinct pair
$j,k$ in Eq. (\ref{eq:bellandr}).

A  tighter version  of monoandry  as  applicable to  quantum boxes  is
reported in  Ref.  \cite{SCH16}, but the  bound (\ref{eq:bellandr}) is
applicable  to any  two-input-two-output  nonsignaling  boxes, and  is
suitable  for our  purpose.  Since  $0\le\mathcal{B}_j\le4$, monoandry
(\ref{eq:bellandr})   implies   that   $-4  \le   \mathcal{B}_j(P)   -
\mathcal{B}_k(P)  \le   4$,  from   which  it  follows   that  $\left|
\mathcal{B}_j(P)   -  \mathcal{B}_k(P)\right|   \le  4$.    Therefore,
$0\le\Gamma\le4$.

If  a   box  $P$  is  nonlocal,   then  there  is  a   $j$  such  that
$\mathcal{B}_j(P)>2$.   By  monoandry,  $\mathcal{B}_k(P)<2$  for  all
$k\ne  j$,  so  that  none  of  the  $\Gamma_j$  will  vanish  in
Eq. (\ref{gi}).   Therefore, $\Gamma(P)>0$, entailing  that boxes
with vanishing $\Gamma$ are necessarily local.  But not all local
boxes satisfy $\Gamma=0$.  This arises from the fact that the set
of boxes characterized by the  property $\Gamma=0$ is not convex,
unlike the local polytope.  In particular, the convex property
\begin{equation}
\Gamma\left(\sum_j p_j P_j\right) \le 
\sum_j p_j\Gamma\left(P_j\right),
\label{eq:g}
\end{equation}
fails when the  correlations $P_j$ in Eq.   (\ref{eq:g}) correspond to
the local-deterministic  boxes. For  these, $\Gamma=0$,  but then
even  local   boxes  arising   from  noncommuting   measurements  have
nonvanishing $\Gamma$, according to the following result.

\begin{thm}\label{thm3}
Locally commuting projective  measurements entail that $\Gamma=0$
for any two-qubit state.
\end{thm}
\begin{proof}
Any two-qubit state, up to local unitary equivalence, can be represented as
\cite{Luo08} 
\begin{eqnarray}
 \rho_{AB}&=&\frac{1}{4}(\openone_A\otimes\openone_B+
 \vec{r}\cdot\vec{\sigma}\otimes\openone_B+\openone_A\otimes\vec{s}\cdot\vec{\sigma}
 \nonumber
 \\ &&+\sum\limits_{i=1}^{3}c_{i}\sigma_i\otimes\sigma_i), \label{a2q}
\end{eqnarray}
where  the  coefficients  $c_i=\tr(\rho_{AB}\sigma_i\otimes\sigma_i)$,
$i=x,y,z$,   form   a  diagonal   matrix   denoted   by  $C$.    Here,
$|\vec{r}|^2+|\vec{s}|^2+||C||^2\le3$  with equality  holding for  the
pure states. The expectation value of the above states is given by,
\begin{equation}
\braket{A_xB_y}=\hat{a}_x\cdot C\hat{b}_y.
\end{equation}
Let  us   calculate  $\Gamma$   for  the   states  as   given  in
Eq. (\ref{a2q}) for commuting measurements on Alice's side. Suppose we
choose measurement  directions as  $\hat{a}_0=\hat{a}_1=\hat{a}$, then
the measurement  observables commute,  i.e., $[A_0,A_1]=0$.   For this
choice  of  commuting  measurements  on Alice's  side,  the  state  in
Eq.      (\ref{a2q})     has      $\mathcal{B}_{0}=\mathcal{B}_{1}=2
\hat{a}_0\cdot C\hat{b}_0$,  and, $\mathcal{B}_{2}=\mathcal{B}_{3}=2
\hat{a}_0\cdot  C\hat{b}_1$. These  values imply  that $\Gamma=0$
for  any choice  of commuting  measurements  on Alice's  side and  any
choice of commuting/non-commuting measurements on Bob's side.  \hfill
\end{proof}

The above result  means that there exist  product boxes (characterized
by   vanishing  pre-shared   dimension)  that   have  $\Gamma>0$.
Therefore,  if  one  goes  beyond noisy  PR  boxes  then  nonvanishing
$\Gamma$   does   not   entail   superlocality.    Such   nonzero
$\Gamma$ product  boxes do  have quantumness due  to noncommuting
measurements,   which   leads   to    local   randomness,   but   this
nonclassicality is not pointed out by superlocality.  \bla

\section{Conclusions and Discussion}\label{conc}

Nonlocality or  steerability in  the given  correlations (or,  box) in
quantum mechanics or in an  arbitrary convex operational theory can be
characterized in terms of the  communication cost that must supplement
pre-shared randomness  in order  to simulate it.   The question  of an
analogous characterization  of nonclassicality arising  from separable
states is addressed here, and associated with superlocality.

However, it  should be pointed  out that the quantumness  indicated by
superlocality  does  not  detect   all  quantum  discord  states.   In
particular,  consider  classical-quantum or  quantum-classical  states
\cite{PHH08}, which have the form:
\begin{equation}
\rho_{AB} = \sum_{j=0,1} p_j |j\rangle\langle j| \otimes \rho_j.
\label{eq:cq}
\end{equation}
In   the   Bell-CHSH   scenario,   for  Alice   measuring   in   basis
$\{|j\rangle\}$, it is clear that  the resulting box can be simulating
by a probabilistic strategy using dimension 2.  This observation holds
even when  Alice measures in any  other basis (except that  her random
number generator will be possibly be more randomized).

It  follows   that  zero-discord  states,   i.e.   classical-classical
correlations   (corresponding   to    orthogonal   $\rho_j$   in   Eq.
(\ref{eq:cq})), are also  non-superlocal. Therefore, superlocal states
are a subset of states with \textit{quantum-quantum} correlations, and
thus  a  strict  subset  of discordant  states.   This  suggests  that
superlocality does not  encompass all of the  nonclassicality in local
quantum states.

As our approach  applies to boxes rather than  specifically to quantum
states,  it leads  in a  natural way  to nonclassicality  in bipartite
states in an arbitrary convex operational theory.  These consideration
can  be extended  to  tripartite  \cite{Jeb16,Jeb14} and  multipartite
boxes.    In  this   context,   Refs.   \cite{Per12,ASA16}   associate
nonclassicality with the nonsimpliciality  of the state space $\Sigma$
of  such  boxes in  a  probability  theory.   Here our  criterion  for
nonclassicality, as indicated by  superlocality, applies to individual
boxes, rather than $\Sigma$.

\section*{Acknowledgements}

C.J thanks IMSc,  Chennai and IISER Mohali for  financial support, and
Dr.  P.   Rungta for for  many inspiring discussions.  He  also thanks
Prof.  N.  Sathyamurthy,  Dr.  K.  P.  Yogendran,  Dr.  Sibasish Ghosh
and  Dr.   Manik Banik  for  useful  discussions.  The  authors  thank
Dr.  E.  Wolfe  for  helpful   comments  and  references.  SA  and  CJ
acknowledge  support through  the  INSPIRE  fellowship [IF120025]  and
project SR/S2/LOP-08/2013 of the DST, Govt.  of India, respectively.

\bibliography{JEB}
\end{document}